\documentclass[runningheads]{llncs}
\usepackage[T1]{fontenc}
\usepackage{graphicx}
\usepackage{tabularx}
\usepackage{subcaption}
\usepackage{amsmath}
\usepackage{amssymb}
\usepackage{amsfonts}
\newcommand\blfootnote[1]{%
  \begingroup
  \renewcommand\thefootnote{}\footnote{#1}%
  \addtocounter{footnote}{-1}%
  \endgroup
}
\begin{document}
\title{A Multi-agent Market Model Can Explain the Impact of AI Traders in Financial Markets -- A New Microfoundations of GARCH model}
\titlerunning{A New Microfoundations of GARCH model}
\author{Kei Nakagawa\inst{1}\orcidID{0000-0001-5046-8128} \and
Masanori Hirano\inst{2}\orcidID{0000-0001-5883-8250} \and
Kentaro Minami\inst{3} \ and
Takanobu Mizuta\inst{4}\orcidID{0000-0003-4329-0645}}
\authorrunning{K. Nakagawa et al.}
\institute{Nomura Asset Management Co,Ltd., Japan \email{kei.nak.0315@gmail.com} \and Preferred Netqorks, Inc., Japan \email{research@mhirano.jp}\and PayPay Corporation, Japan \email{research@ktrmnm.jp}\and SPARX Asset Management Co., Ltd., Japan \email{mizutata@gmail.com}}

\maketitle              % typeset the header of the contribution
\begin{abstract}
%The abstract should briefly summarize the contents of the paper in 150--250 words.
The AI traders in financial markets have sparked significant interest in their effects on price formation mechanisms and market volatility, raising important questions for market stability and regulation. Despite this interest, a comprehensive model to quantitatively assess the specific impacts of AI traders remains undeveloped. This study aims to address this gap by modeling the influence of AI traders on market price formation and volatility within a multi-agent framework, leveraging the concept of microfoundations.
Microfoundations involve understanding macroeconomic phenomena, such as market price formation, through the decision-making and interactions of individual economic agents. While widely acknowledged in macroeconomics, microfoundational approaches remain unexplored in empirical finance, particularly for models like the GARCH model, which captures key financial statistical properties such as volatility clustering and fat tails. 
This study proposes a multi-agent market model to derive the microfoundations of the GARCH model, incorporating three types of agents: noise traders, fundamental traders, and AI traders. By mathematically aggregating the micro-structure of these agents, we establish the microfoundations of the GARCH model. We validate this model through multi-agent simulations, confirming its ability to reproduce the stylized facts of financial markets. Finally, we analyze the impact of AI traders using parameters derived from these microfoundations, contributing to a deeper understanding of their role in market dynamics.
\blfootnote{This paper does not reflect the view of the organizations the authors belong to. All errors in this paper are the responsibility of the authors.}
\keywords{AI Trader  \and GARCH model \and microfundation \and financial market \and simulation}
\end{abstract}
\section{Introduction}
The advancement of various machine learning algorithms has markedly increased the presence of AI traders in financial markets, drawing significant attention to their potential impact on these markets. Understanding how AI traders influence price formation mechanisms and market volatility is crucial for maintaining market stability and informing regulatory perspectives~\cite{lee2020access,leitner2024rise}. 
However, a comprehensive model to quantitatively capture the specific impacts of AI traders has not yet been fully established. In this study, we aim to model the effects of the presence and decision-making processes of AI traders on market price formation mechanisms and volatility, utilizing a multi-agent modeling framework and leveraging the concept of microfoundations.

The effort to understand macroeconomic phenomena, such as market price formation, through the decision-making and interactions of economic agents, is known as microfoundations. 
In macroeconomics, the significance of microfoundations has been widely recognized, leading to the development of numerous macroeconomic models incorporating microfoundations \cite{janssen2005microfoundations,hoover2015reductionism}. 
Conversely, attempts to provide microfoundations for price fluctuation models in empirical finance are relatively sparse. 
For instance, the GARCH (Generalized AutoRegressive Conditional Heteroscedasticity) model, an extension of the ARCH model~\cite{engle1982autoregressive}, is extensively used in finance to model conditional variance in various empirical studies~\cite{nakagawa2018risk,francq2019garch}. 
This is because the GARCH model effectively reproduces typical financial statistical properties, known as stylized facts~\cite{cont2001empirical}, such as volatility clustering\footnote{This refers to the phenomenon where large~(resp. small) changes in asset prices are often followed by large~(resp. small) changes, indicating that volatility tends to cluster together over time. This results in periods of high volatility and periods of low volatility within financial markets.} and fat tails\footnote{This refers to the statistical property of a probability distribution where the tails (extremes) are fatter than those of a normal distribution. Fat tails are often quantified using kurtosis, a measure of the "tailedness" of the distribution. Distributions with high kurtosis exhibit fat tails, indicating a greater likelihood of observing values far from the mean.}. 
While \cite{mizuta2016micro} provides microfoundations for the ARCH model, the microfoundations for the GARCH model remain underexplored. Specifically, it is unclear how the parameters of the GARCH model arise from the micro-level decision-making processes or the proportions of different investors.
The validity of multi-agent modeling is often evaluated based on its ability to reproduce these stylized facts \cite{chen2012agent}.

Our objective is to construct the GARCH model through microfoundations and ensure that our model can reproduce the stylized facts of financial markets represented by the GARCH model.
The validity of multi-agent modeling is evaluated based on whether they can reproduce the stylized facts\cite{chen2012agent}. 
We aim to construct the GARCH model throgh microfoundations and ensure that our model can reproduce the stylized facts of financial markets represented by the GARCH model.

Based on the above, we propose a multi-agent market model that can analyze the impact of AI traders on financial markets by deriving the microfoundations of the GARCH model. 
Our multi-agent market model includes three types of agents: noise traders, fundamental traders, and AI traders. 
By mathematically solving the aggregation of the micro-structure of the multi-agent market model, we demonstrate the microfoundations of the GARCH(1,1) model. 
Furthermore, through multi-agent simulations, we validate the multi-agent market model and confirm the replication of market stylized facts. 
Finally, we discuss the impact of AI traders using parameters derived from the microfoundations.

\section{Related Work}
Research on the impact of traders on financial markets has increased significantly with the rapid expansion of new technologies such as algorithmic trading and AI. 
Various approaches have been employed to evaluate the influence of traders on the market. We review key prior studies on the impact of traders, their analysis using multi-agent models, and the perspective of microfoundations, highlighting the positioning of this study.

This study aims to bridge the gap in understanding the impacts of AI traders on market price formation and volatility by leveraging a multi-agent modeling approach grounded in microfoundations. Previous research has extensively explored high-frequency trading and various trader types, but the specific quantitative impacts of AI traders remain underexplored. Our work addresses this gap by developing a novel multi-agent market model that includes noise traders, fundamental traders, and AI traders, capturing their decision-making processes and interactions to analyze their collective influence on financial markets.

We develop the concept of microfoundations to the GARCH(1,1) model, demonstrating how its parameters can arise from the micro-level behaviors of different traders. This bridges a crucial gap in empirical finance, where the microfoundations of the GARCH model have been underexplored. 
By providing new insights into how AI traders affect market volatility and price formation, our research offers valuable information for market stability and regulatory frameworks. This study advances the theoretical understanding of AI traders' impact and addresses the existing gap in empirical finance, offering practical implications for market participants and policymakers.

\subsection{Impact of Traders on the Market}
The initial research on the impact of traders primarily focused on high-frequency trading~(HFT). 
Hendershott {\it et al.} \cite{hendershott2011does} demonstrated that HFT improves market liquidity and reduces trading costs. However, it has also been suggested that HFT could increase market volatility \cite{brogaard2010high}. These studies imply that the impact of traders on the market is not straightforward.

Other research has explored the broader impact of various types of traders. 
For instance, Kirilenko {\it et al.} \cite{kirilenko2017flash} analyzed the role of different trader types in the flash crash of 2010, showing how their interactions contributed to extreme market events. These studies highlight the complex dynamics and potential systemic risks posed by different trading strategies.

\subsection{Multi-Agent Models for Market Analysis}
To analyze the impact of traders, multi-agent models have been widely used. 
Lebaron {\it et al.} \cite{lebaron2001evolution} studied the effect of interactions among agents with different trading strategies on market dynamics using an artificial market. This research demonstrated the utility of agent-based models in capturing the emergent properties of financial markets.
Chen {\it et al.} \cite{chen2012agent} employed an agent-based model to evaluate the influence of algorithmic traders on market price formation and volatility. Their work showed how the inclusion of algorithmic traders can lead to different market behaviors and identified key factors that influence these outcomes. These studies provide a foundational understanding of how agent-based modeling can be applied to financial markets to assess the impact of traders.

\subsection{Microfoundations Perspective}
Farmer {\it et al.} \cite{farmer2002price} provided a microfoundational approach to price formation by modeling the behavior of heterogeneous agents and their impact on market prices. Similarly, Lux {\it et al.} \cite{lux1999scaling} developed a model that captured the scaling laws observed in financial markets by incorporating the interactions of traders with different strategies. 
They argued that heterogeneous agents and their interactions are necessary to replicate stylized facts through a multi-agent simulation study.
Because they utilized simulation rather than mathematical analysis, their methods are a bit different from those of the microfoundation, but their motivation is very similar.
These studies illustrate the potential of microfoundations to explain observed market phenomena through the lens of individual agent behavior.

\section{Multi-agent Market Model}
In this section, we define the decision-making model for each trader and the price determination process used in our multi-agent market model.

\subsection{Return Model}
The return model $r_t$ we use in this study is GARCH(1,1) \cite{bollerslev1986generalized}, which is widely used in financial markets analysis and defined as: 
\begin{align}
    & r_t = f(x_{t-1}, I_{t-1}) + u_t \\
    & u_t = \sigma_t \varepsilon_t,~ \varepsilon_t \sim N(0,1) \\
    & \sigma_t^2 = \omega + \alpha u_{t-1}^2 + \beta \sigma_{t-1}^2 \label{GARCH}
\end{align}
where $r_t$ is the return at time $t$, $x_{t-1}$ and $I_{t-1}$ are the fundamental variable and the set of information (e.g., historical return, sentiment indicators) available for prediction up to time $t-1$.

$f$ is an appropriate function for explaining future returns based on fundamental values and exogenous information.

This setting is consistent with the semi-strong form of efficient market hypothesis~\cite{fama1970efficient}, i.e., the assumption that all publicly available information is taken into account in the asset return.
$u_t = r_t - f(x_{t-1}, I_{t-1})$ is the residual term.
Furthermore, $\omega >0$ and $\alpha + \beta < 1$ are necessary for satisfying the stationarity of the GARCH model.
$\omega$ and $\alpha$ mean the constant volatility and the sensitivity to past return shocks, respectively.
$\beta$ is the dependence on past volatility,  which is necessary to replicate volatility clustering.

The expected returns 
$\mathbb{E}_{t-1}[r_t]$ and the variance $\mathbb{V}_{t-1}[r_t]$ are calculated as:
\begin{align}
    &\mathbb{E}_{t-1}[r_t] = \mathbb{E}_{t-1}[f(x_{t-1}, I_{t-1})] + \mathbb{E}_{t-1}[u_t] = f(x_{t-1}, I_{t-1})\label{GARCH_EXP} \\
    &\mathbb{V}_{t-1}[r_t]:=\mathbb{E}_{t-1}[(r_t - \mathbb{E}_{t-1}[r_t])^2] = \omega + \alpha u_{t-1}^2 + \beta \sigma_{t-1}^2 \label{GARCH_VAR}.
\end{align}

Hereafter, $\mathbb{E}_{t-1}[\cdot]$ is the conditional expectation and $\mathbb{V}_{t-1}[\cdot]$ is the conditional variance up to time $t-1$.

\subsection{Trader Model}
We assume the existence of three types of traders in our market model:
\begin{description}
    \item[1. Noise traders] who provide liquidity,
    \item[2. Fundamental traders] who make decisions based on fundamental variables,
    \item[3. AI traders] who make decisions based on the model predictions learned from past market time series.
\end{description}

\subsubsection{Noise Traders}
Noise traders are based on the existence of traders who generate orders probabilistically due to exogenous reasons \cite{peress2020glued}.
These exogenous reasons include investment adjustments over individual traders' life cycles, portfolio rebalancing for hedging purposes, and passive investment strategies that follow indices, such as indices weighted by market capitalization.

We assume that noise traders place orders uniformly, irrespective of the fundamental prices. Consequently, noise traders play the role of providing liquidity.

\subsubsection{Fundamental Traders}
Next, we assume the existence of risk-averse fundamental traders who make decisions based on the fundamental variable $x_{t-1}$ at the previous time step $t-1$.
Specifically, we assume the utility function $U(\cdot)$ of the returns for the fundamental traders $r_{F,t}$ at the time step $t$.
These traders are risk-averse investors, meaning that their utility function satisfies $U' > 0$ and $U'' < 0$.

Here, we employ the utility function $U$ depending on the mean and standard deviation (not variance) of their returns $r_{F,t}$.
Furthermore, we assume that:
\begin{align}
	\mathbb{E}_{t-1}[r_{F,t}] &= g(x_{t-1})\\
	\mathbb{V}_{t-1}[r_{F,t}] &= \sigma_{t-1}^2
\end{align}
where $g$ is a monotonic function and $\sigma_{t-1}$ is defined previous section. % and $\sigma_{t-1}$ is newly defined here.
That means, when fundamental traders obtain information at time $t-1$, they generally suppose a monotonically expected return on the fundamental variable, and the variance of this return is $\sigma_{t-1}^2$ available at time $t-1$.
Thus, the expected utility $U$ of the fundamental traders can be written as follows:
\begin{align}\label{U_Funda_trader}
    \mathbb{E}_{t-1}[U(r_{F,t})] = \mathbb{E}_{t-1}[r_{F,t}] - \lambda \sqrt{\mathbb{V}_{t-1}[r_{F,t}]} = g(x_{t-1}) - \lambda \sigma_{t-1}
\end{align}
where the risk aversion coefficient is $\lambda > 0$.

In the field of finance, the risk-averse expected utility is generally described by the mean and variance\cite{markowitz2000mean}. 
Note that for fundamental traders in this study, utility function is based on the mean and standard deviation. 
In general, the use of the mean and standard deviation instead of the mean and variance for the risk-averse utility function can be justified under certain conditions \cite{tobin1965theory}.
Specifically, (1) if the utility function is a quadratic function of $r$, or (2) if the probability distribution of $r$ is normal. However, these are quite strong assumptions. In particular, since the aim of this study is to reproduce stylized facts, normality cannot be required.
Therefore, in this paper, we assume the following, which is a more relaxed condition that includes the normal distribution:
The standardized probability distribution of $r_{F,t}$ is independent of the values of the expected value and the standard deviation.
The lemma \ref{ref:lemma} justifies this condition.

\begin{lemma}\label{ref:lemma}
If the standardized probability distribution of the random variable $r$ is independent of the values of the expected value $\mu$ and the standard deviation $\sigma$, i.e., if the probability distribution $\phi(\epsilon)$ of $\epsilon = (r - \mu) / \sigma $ does not depend on $\mu$ and $ \sigma $, then an increase in the standard deviation results in a decrease in the expected utility (allowing the use of the standard deviation as a risk measure).
\end{lemma}
\begin{proof}
    The proof is given in the Appendix.
\end{proof}

\subsubsection{AI Traders}
AI traders are assumed to train each own prediction model $h$ that effectively explains the past return series $r_{t}, t=1, \cdots, T$.
We assumed that they also have a utility $M(\cdot)$ based on the model.

That is, under a certain loss function $L$, the model $h$ is trained by minimizing the following:
\begin{align}
    \min_{h} \sum_{t} L(r_{t},h(x_{t-1}, I_{t-1}))
\end{align}
where $I_{t-1}$ is the set of information available for prediction up to time $t-1$ as the same as that in the previous section.

In this study, we use the squared loss function $L(x,y) = (x-y)^2$, which is commonly used in regression tasks.
Therefore, we minimize:
\begin{align}
    \min_{h} \sum_{t} (r_{t}-h(x_{t-1}, I_{t-1}))^2
\end{align}
to train the model $h$.

Next, let the utility function of the returns $r_{A,t}$ for the AI traders be denoted by $M(\cdot)$.
These traders are assumed to be risk-averse investors, meaning that $M' > 0$ and $M'' < 0$.
Similar to fundamental traders, we assume that the $M$ of $r_{A,t}$ is determined by its mean and standard deviation.
Furthermore, we assume that $r_{A,t}$ is given by the model $h$ as follows:
\begin{align}
    &\mathbb{E}_{t-1}[r_{A,t}] = h(x_{t-1}, I_{t-1})\\
    &\mathbb{V}_{t-1}[r_{A,t}] = \mathbb{E}_{t-1}\left[(r_{t-1}-\mathbb{E}_{t-1}[h(x_{t-2}, I_{t-2})])^2\right]
\end{align}
This means when AI traders obtain available information until time $t-1$, they estimate an expected return based on the model.
Moreover, the risk is assumed to be the prediction error (standard deviation) of the model available at time $t-1$: $\sqrt{\mathbb{V}_{t-1}[r_{A,t}]}$.
Thus, the expected utility of AI traders can be written as:
\begin{align}\label{U_AI_trader_original}
    \mathbb{E}_{t-1}[M(r_{A,t})] &= \mathbb{E}_{t-1}[r_{A,t}] - \gamma \sqrt{\mathbb{V}_{t-1}[r_{A,t}]} \\
    & = h(x_{t-1}, I_{t-1}) - \gamma \mathbb{E}_{t-1}\left|r_{t-1} - \mathbb{E}_{t-1}[h(x_{t-2}, I_{t-2})]\right|
\end{align}
where $\gamma (> 0)$ is a risk aversion coefficient.
The utility is based on the prediction value, penalized for the prediction deviation from the model $h$.

\subsection{Pricing Model}
Our model employs a simple price model based on the order imbalance, which is proposed in a previous work \cite{Mizuta2012}.

In the price model, the price return $r_t$ based on the order imbalance is calculated as:
\begin{align}\label{pricing}
    r_t = \rho \frac{A^b_t - A^s_t}{A^b_t + A^s_t}
\end{align}
where $A^b_t$ is the total buy orders, $A^s_t$ is the total sell orders, and $\rho$ is the constant coefficient of the order imbalance.

Moreover, $A^b_t,A^s_t$ are defined as:
\begin{align}
     A^b_t &= \frac{1}{2} S\left(1 +p_1 \mathbb{E}_{t-1}[U(r_{F,t})] +p_2 \mathbb{E}_{t-1}[M(r_{A,t})]\right) \nonumber\\
    &+ \frac{1}{2}k S(1+ p_1 \mathbb{E}_{t-1}[U(r_{F,t})]+ p_2 \mathbb{E}_{t-1}[M(r_{A,t})])\varepsilon_t \label{Abt}\\
     A^s_t &= \frac{1}{2} S\left(1 -p_1 \mathbb{E}_{t-1}[U(r_{F,t})] -p_2 \mathbb{E}_{t-1}[M(r_{A,t})]\right) \nonumber\\
    &- \frac{1}{2}k S(1+ p_1 \mathbb{E}_{t-1}[U(r_{F,t})]+ p_2 \mathbb{E}_{t-1}[M(r_{A,t})])\varepsilon_t \label{Ast}
\end{align}
where $S,k$ are constant variables representing the liquidity, $p_1$ and $p_2$ are the ratios of fundamental and AI traders over noise traders, respectively.
$p_1 + p_2$ indicates the ratio of investors who take liquidity from the market to those who provide liquidity (noise traders).
Additionally, $\mathbb{E}_{t-1}[U(r_{F,t})]$ and $\mathbb{E}_{t-1}[M(r_{A,t})]$ are the utility functions of fundamental and AI traders defined by equations \eqref{U_Funda_trader} and \eqref{U_AI_trader}, respectively.
The first term represents the constant total order quantities placed by noise, fundamental, and AI traders.
The second term represents the difference in buy and sell orders due to demand-supply imbalances.
The constant $k$ indicates the order quantities generated due to demand-supply imbalances relative to the constant total order quantities, which highly affects how much demand and supply imbalances take liquidity from the market.
A larger value of $k$ implies lower liquidity.

\section{Theoretical Analysis: The Microfoundations of GARCH}

In this section, we derive the microfoundations of GARCH based on the multi-agent market model explained in the previous section.
The final goal of this section is to derive the GARCH parameters ($f, \omega, \alpha, \beta$) in equation \eqref{GARCH} using the parameters defined in the previous section.

First, $\mathbb{E}_{t-1}[h(\cdot)] = f(\cdot)$ is reasonably assumed because the effort to find well-predictable model $h$ is the same as the finding the ground truth function $f$.
Thus, equation \eqref{U_AI_trader_original} can be re-formulated as:
\begin{align}\label{U_AI_trader}
    \mathbb{E}_{t-1}[M(r_{A,t})] &= h(x_{t-1}, I_{t-1}) - \gamma \mathbb{E}_{t-1}\left|r_{t-1} - \mathbb{E}_{t-1}[h(x_{t-2}, I_{t-2})]\right|\\
    & = h(x_{t-1}, I_{t-1}) - \gamma \left|r_{t-1} - f(x_{t-2}, I_{t-2})\right| \\
    & = h(x_{t-1}, I_{t-1}) - \gamma \left|u_{t-1}\right|
\end{align}

Since the total volume of buy and sell orders is given by the equations \eqref{Abt} and \eqref{Ast}, substituting them into the equation \eqref{pricing} in the pricing model as:
\begin{align}
r_t &= \rho (p_1 \mathbb{E}_{t-1}[U(r_{F,t})] + p_2 \mathbb{E}_{t-1}[M(r_{A,t})]) \nonumber\\
&+ \rho k \left\{1+p_1 \mathbb{E}_{t-1}[U(r_{F,t})] + p_2 \mathbb{E}_{t-1}[M(r_{A,t})] ]\right\}\varepsilon_t\\
&= \rho (p_1 (g(x_{t-1}) - \lambda \sigma_{t-1}) +p_2 ( h(x_{t-1}, I_{t-1}) - \gamma |u_{t-1}|)) \nonumber\\
&+ \rho k \{1+ p_1 (g(x_{t-1}) - \lambda \sigma_{t-1}) + p_2 ( h(x_{t-1}) - \gamma |u_{t-1}|))\}\varepsilon_t.
\end{align}
Therefore, $\mathbb{E}_{t-1}[r_t]$ and $\mathbb{V}_{t-1}[r_t]$ can be derived as:
\begin{align}
    \mathbb{E}_{t-1}[r_t] &= \rho (p_1 (g(x_{t-1}) - \lambda \sigma_{t-1}) +p_2 ( h(x_{t-1}) - \gamma |u_{t-1}|)), \\
    \mathbb{V}_{t-1}[r_t] &= \rho^2 k^2(1 + p_1^2 g(x_{t-1})^2 + p_2^2 h(x_{t-1})^2) + \rho^2 k^2 p_2^2 \gamma^2 u_{t-1}^2 + \rho^2 k^2 p_1^2 \lambda^2 \sigma_{t-1}^2.
\end{align}
Compared with equations \eqref{GARCH_EXP} and \eqref{GARCH_VAR} in the GARCH(1,1) model, we can prove the following theorem.

\begin{theorem}[Microfoundation of GARCH Model]\label{thm1}
The parameters of the equation \eqref{GARCH} are given by the microfoundation as:
\begin{align}
&\omega=\rho^2 k^2\{1 + p_1^2 g(x_{t-1})^2 + p_2^2 h(x_{t-1}, I_{t-1})^2\},\label{all_omega}\\
&f=\rho \{p_1 (g(x_{t-1}) - \lambda \sigma_{t-1}) + p_2 (h(x_{t-1}, I_{t-1}) - \gamma |u_{t-1}|)\}, \label{all_f}\\
&\alpha=\rho^2 k^2 p_2^2 \gamma^2, \quad \beta=\rho^2 k^2 p_1^2 \lambda^2 \label{all_ab}
\end{align}
\end{theorem}
Theorem \ref{thm1} provides a significant contribution to understanding the microfoundations of the GARCH(1,1) model by linking it to the behaviors and interactions of different types of traders in a financial market. Several key insights and implications can be drawn from this theorem:

In a market where noise, fundamental, and AI traders exist, there are 
\begin{description}
    \item[Constant volatility($\omega \neq 0$):] The parameter $\omega$ represents the constant component of volatility in the GARCH model. According to the theorem, $\omega$ is influenced by the aggregate behaviors of noise, fundamental, and AI traders. This implies that even in the absence of external shocks, the inherent trading activities of these agents contribute to a baseline level of market volatility. 
    \item[Incorporation of information($f \neq 0$):] The parameter $f$ in the theorem indicates that fundamental variables and available information up to time $t-1$ are incorporated into market returns. This aligns with the semi-strong form of the efficient market hypothesis, suggesting that market prices reflect all publicly available information.
    \item[Sensitivity to past shocks($\alpha \neq 0$):] The parameter $\alpha$ signifies the market's sensitivity to past price shocks. The theorem shows that $\alpha$ is directly related to the trading activities of AI traders, particularly their risk aversions. The existence of AI traders is also amplifying market responses to past events.
    \item[Volatility clustering($\beta \neq 0$):]The parameter $\beta$ represents the persistence of volatility over time. The theorem links $\beta$ to the dependence on past volatility, driven by fundamental traders' risk aversions. 
    The volatility clustering is a crucial characteristic of financial markets, and this theorem provides a microfoundational explanation for this phenomenon.
\end{description}

Below, we derive the results of the microfoundations of the GARCH model when we assume the non-existence of each trader based on Theorem \ref{thm1}.

\subsubsection{When only the noise traders exist in markets ($p_1=0,p_2=0$),} Theorem \ref{thm1} can be reformulated as:
\begin{align}
    &\omega=\rho^2 k^2, \label{n_omega} \\ 
    &f=\alpha=\beta=0 \label{n_fab}.
\end{align}
We can derive the following proposition:
\begin{proposition}[Microfoundation of the GARCH model in a market with only noise traders]
When only the noise traders exist in markets ($p_1=0,p_2=0$), there are 

\begin{description}
    \item[Constant volatility($\omega \neq 0$):] The volatility $\omega$ depends on the variance of the order imbalance ($\rho$) and liquidity ($k$).    
    \item[Exclusion of information($f = 0$):] 
    The parameter $f$ being zero indicates that information has no impact on market prices when only noise traders are present. This contrasts with markets where fundamental and AI traders are active, suggesting that the presence of noise traders alone does not incorporate information into asset pricing.
    \item[Absence of response to past shocks($\alpha = 0$):] The parameter $\alpha$ being zero signifies that the market does not exhibit sensitivity to past price shocks when only noise traders are present. 
    \item[No Volatility Clustering($\beta = 0$):]The parameter $\beta$ being zero indicates the absence of volatility clustering in a market with only noise traders because noise traders' actions do not create persistent volatility patterns. 
\end{description}
\end{proposition}

\subsubsection{When the noise and fundamental traders exist in markets ($p_1 > 0,p_2=0$),} Theorem \ref{thm1} can be reformulated as:
\begin{align}
&\omega=\rho^2 k^2(1 + p_1^2 g(x_{t-1})^2),\label{nf_omega} \\
&f=\rho p_1 (g(x_{t-1}) - \lambda \sigma_{t-1}), \label{nf_f} \\
&\alpha=0,\beta=\rho^2 k^2 p_1^2 \lambda^2 \label{nf_ab}.
\end{align}
This leads to the following proposition:

\begin{proposition}[Microfoundation of the GARCH model in a market with the noise and fundamental traders]
When the noise and fundamental traders exist in markets ($p_1 > 0,p_2=0$), the following are satisfied:
\begin{description}
    \item[Constant volatility($\omega \neq 0$):] The parameter $\omega$ is influenced by the order imbalance ($\rho$) and liquidity ($k$), as well as the trading activities of fundamental traders.
    \item[Incorporation of Fundamental Information($f \neq 0$):]
    The presence of fundamental traders ensures that fundamental variables are reflected in market prices. This is crucial for market efficiency.
    \item[No Response to Past Shocks($\alpha=0$):]
    The absence of AI traders means that the market does not exhibit sensitivity to past price shocks. 
    \item[Volatility Clustering($\beta \neq 0$):]
    The parameter $\beta$ being non-zero indicates that volatility clustering occurs in markets with noise and fundamental traders. The presence of fundamental traders, who react to fundamental variables, contributes to this persistent volatility pattern.
\end{description}    
\end{proposition}

\subsubsection{When the noise and AI traders exist in markets ($p_1 = 0,p_2 > 0$),} theorem \ref{thm1} can be reformulated as:
\begin{align}
&\omega=\rho^2 k^2(1 + p_2^2 h(x_{t-1}, I_{t-1})^2),\label{na_omega}\\ 
&f=\rho p_2 (h(x_{t-1}, I_{t-1}) - \gamma |u_{t-1}|), \label{na_f}\\
&\alpha=\rho^2 k^2 p_2^2 \gamma^2,\beta=0 \label{na_ab}
\end{align}
We can derive the following proposition:
\begin{proposition}[Microfoundation of the GARCH model in a market with the noise and AI traders]
When the noise and AI traders exist in markets ($p_1 = 0,p_2>0$), the following are satisfied:
\begin{description}
\item[Constant volatility ($\omega \neq = 0$):] The parameter $\omega$ is influenced by the order imbalance ($\rho$) and liquidity ($k$), as well as the presence of AI traders.
\item[Incorporation of Information ($f \neq 0$):] The parameter $f$ being non-zero indicates that the information processed by AI traders is incorporated into market prices. This suggests that the trading strategies of AI traders ensure that their predictive information is reflected in asset pricing.
\item[Response to Past Shocks ($\alpha \neq 0$):] The parameter $\alpha$ being non-zero signifies that the market exhibits sensitivity to past price shocks due to the presence of AI traders. This shows the role of AI traders in amplifying market responses to historical price movements.
\item[No Volatility Clustering ($\beta=0$):] The parameter 
$\beta$ being zero indicates the absence of volatility clustering in a market with only noise and AI traders because their actions do not create persistent volatility patterns.
\end{description}
\end{proposition}

\section{Simulation Experiment}
\subsection{Experimental Settings}
In the previous section, we discuss the mathematical structure without a specific definition for $g$ and $h$ for fundamental and AI traders.
In this section, we temporarily define those functions and other parameters and check how the numerical simulation works well by checking stylized facts, which are well-known phenomena in financial markets.

For the parameters, we employed $\rho=4$, $k=0.4$, $p_1=0.2$, $p_2=0.4$, $\lambda =1.2$, $\gamma=1.2$.
Those parameters are empirically decided, and $p_1$ and $p_2$ are set to satisfy the situation that enough liquidity traders exist in the simulation.
The fundamental variable $x_t$ was generated by a random variable following a standard normal distribution.
The log utility function is employed for $g$:
\begin{align}
	g(x_{t-1})=\log(1 + \max(-0.99,x_{t-1})).
\end{align}
For $h$, the auto-regressive~(AR) process is assumed:
\begin{align}
	h(x_{t-1}, I_{t-1})= 0.1 x_{t-1}.
\end{align}

We simulate the return series $\hat{r}_t$ based on equation \eqref{pricing}.
We evaluate the validity of the multi-agent market model based on whether it can reproduce stylized facts \cite{chen2012agent}.

We evaluate several stylized facts using $\hat{r}_t$: negative skewness, higher kurtosis than the normal distribution~(leptokurtic), non-normality~(Kolmogorov-Smirnov; KS test), and positive autocorrelation of the square of returns~(volatility clustering). 
The generated time series length is denoted by $\hat{T}=1000$.

The skewness is defined as:
\begin{align}
    \text{Skewness} = \frac{\frac{1}{\hat{T}} \sum_{t=1}^{\hat{T}} (\hat{r}_t - \bar{r})^3}{\left(\frac{1}{\hat{T}} \sum_{t=1}^{\hat{T}} (\hat{r}_t - \bar{r})^2\right)^{3/2}}
\end{align}
where $\bar{r}$ is the mean of the return series. 
A negative skewness value indicates a distribution with a long tail on the left side~(the loss side), suggesting that small losses occur frequently along with occasional large losses.

The kurtosis is defined as:
\begin{align}
    \text{Kurtosis} = \frac{\frac{1}{\hat{T}} \sum_{t=1}^{\hat{T}} (\hat{r}_t - \bar{r})^4}{\left(\frac{1}{\hat{T}} \sum_{t=1}^{\hat{T}} (\hat{r}_t - \bar{r})^2\right)^2}
\end{align}
A high kurtosis value means that the distribution of returns is sharper than a normal distribution, indicating that values near the center~(mean) are concentrated, and extreme values~(very high returns or very large losses) occur more frequently than expected.

The KS test evaluates whether the return series follows a normal distribution. The KS test statistic is defined as:
\begin{align}
    \text{KS statistics} = \sup_x | F_{\hat{r}_t}(x) - F_{\text{normal}}(x) |
\end{align}
where $F_{\hat{r}_t}(x)$ is the empirical cumulative distribution function~(CDF) of $\hat{r}_t$ and $F_{\text{normal}}(x)$ is the CDF of the normal distribution with the same mean and variance as $\hat{r}_t$. 
The KS test would reject the null hypothesis if the returns exhibit non-normal characteristics.

The $\tau$-order autocorrelation of the square of returns is defined as:
\begin{align}
    \text{Autocorrelation}(\tau) = \frac{\sum_{t=\tau+1}^{\hat{T}} (\hat{r}_t^2 - \bar{r}^2)(\hat{r}_{t-\tau}^2 - \bar{r}^2)}{\sum_{t=1}^{\hat{T}} (\hat{r}_t^2 - \bar{r}^2)^2}
\end{align}
for lag $tau$. 
This measure indicates the presence of volatility clustering, where large price changes are likely to be followed by other large price changes.
The 1st-order autocorrelation is defined by $\tau = 1$.

These stylized facts are derived from previous works, such as \cite{Cont2001}, and provide a comprehensive evaluation of the multi-agent market model's ability to reproduce real market behaviors.

\subsection{Simulation Results}
\begin{table}[]
    \centering
    \caption{Statistics on generated return series. $^{*}$p$<$0.1; $^{**}$p$<$0.05; $^{***}$p$<$0.01.}
    \label{stat}
    \begin{tabular}{ll}
    \hline
        Skewness ($<0$) &  -1.880*** \\ 
        Kurtosis ($>3$) &  9.230*** \\ 
        KS statistics & 0.788*** \\ 
        1st-order autocorrelation ($>0$) &  0.367*** \\ \hline
    \end{tabular}
\end{table}

Table \ref{stat} summarizes the statistics for the generated return series, with significance levels indicated by asterisks. 

The results from Table \ref{stat} demonstrate that the multi-agent market model successfully generates return series that exhibit key stylized facts observed in real financial markets. 
The significant negative skewness and high kurtosis indicate that the generated returns can capture the asymmetric and heavy-tailed nature of financial returns. The significant KS statistic further confirms that the generated returns do not follow a normal distribution, aligning with empirical observations in financial data.
The positive autocorrelation of the square of returns suggests that the generated returns effectively replicate volatility clustering.

The sample paths of the generated return series are shown in Figure \ref{timeseries}. The histograms of returns are also shown in Figure \ref{hist}. These visual representations further support the statistical findings summarized in Table \ref{stat}. The distribution of returns exhibits a non-normal, heavy-tailed property, with noticeable skewness to the left, indicating frequent small losses and occasional large losses. The return series demonstrates periods of high volatility interspersed with periods of relative calm, reflecting the volatility clustering.

\begin{figure}[htbp]
 \centering
 \begin{subfigure}[b]{0.58\textwidth}
  \centering
  \includegraphics[height=4cm]{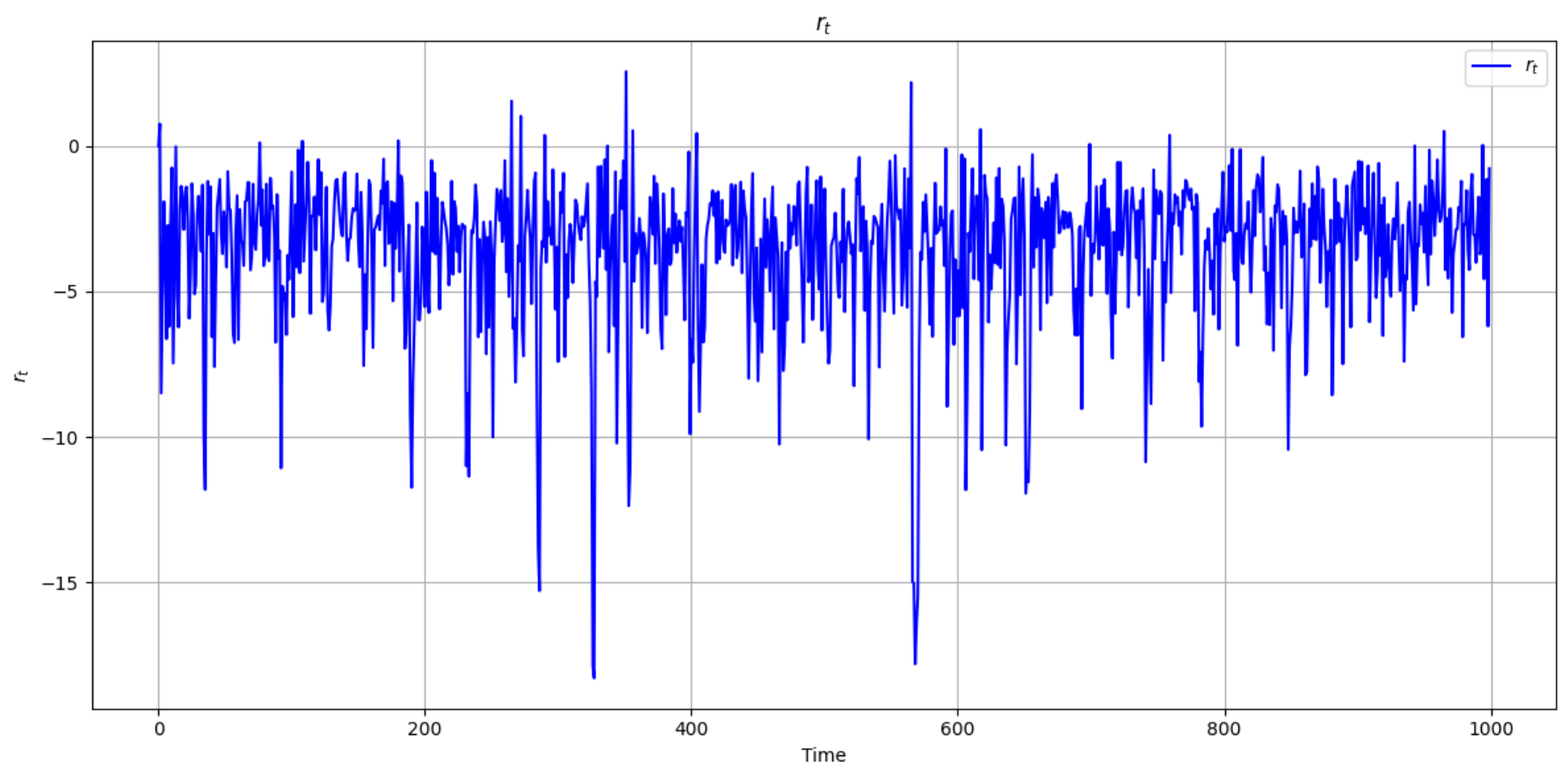}
  \caption{Generated return time-series}
  \label{timeseries}
 \end{subfigure}
 \hfill
 \begin{subfigure}[b]{0.38\textwidth}
  \centering
  \includegraphics[height=4cm]{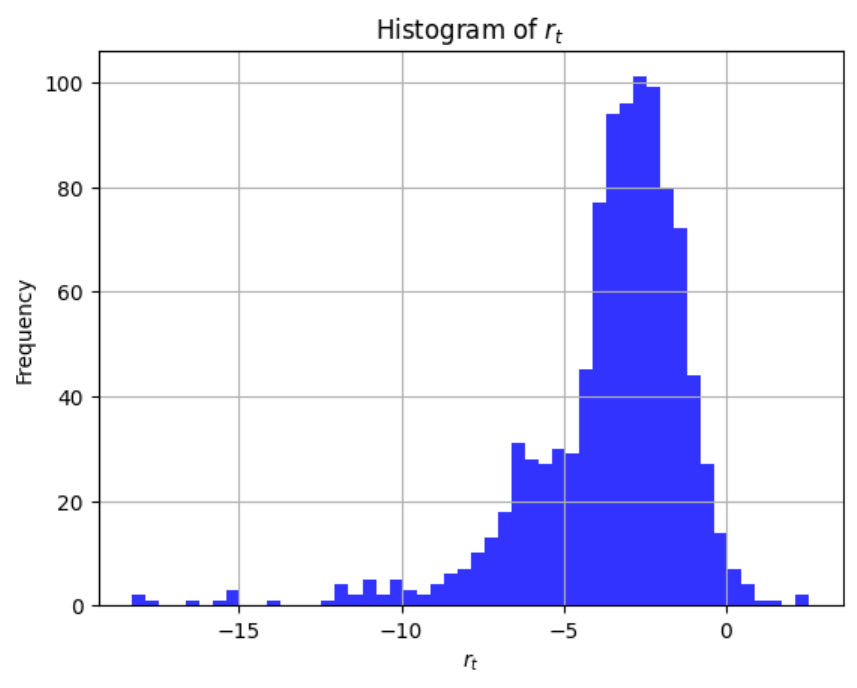}
  \caption{The return distribution of the generated return series}
  \label{hist}
 \end{subfigure}
 \caption{Comparison of generated return time-series and their distribution}
 \label{fig:comparison}
\end{figure}

\section{Discussion}
\begin{table}[h]
\centering
\caption{Comparison of Market with Different Trader Types}
\begin{tabularx}{\textwidth}{>{\raggedright\arraybackslash}X>{\centering\arraybackslash}X>{\centering\arraybackslash}X>{\centering\arraybackslash}X>{\centering\arraybackslash}X}
\hline
\textbf{Market Type} & \textbf{Baseline Volatility ($\omega$)} & \textbf{Incorporation of Information ($f$)} & \textbf{Response to Past Shocks ($\alpha$)} & \textbf{Volatility Clustering ($\beta$)} \\ \hline
Only Noise Traders (\(p_1 = 0, p_2 = 0\)) & Eq. $\eqref{n_omega}$ & - & - & - \\ \hline
Noise and Fundamental Traders (\(p_1 > 0, p_2 = 0\)) & Eq. $\eqref{nf_omega}$ & Eq. $\eqref{nf_f}$ & - & Eq. $\eqref{nf_ab}$  \\ \hline
Noise and AI Traders (\(p_1 = 0, p_2 > 0\)) & Eq. $\eqref{na_omega}$ & Eq. $\eqref{na_f}$ & - & Eq. $\eqref{na_ab}$\\ \hline
Noise, Fundamental, and AI Traders (\(p_1 > 0, p_2 > 0\)) & Eq. $\eqref{all_omega}$ & Eq. $\eqref{all_f}$ & Eq. $\eqref{all_ab}$ & Eq. $\eqref{all_ab}$ \\ \hline
\end{tabularx}
\label{tbl:market_type}
\end{table}

Firstly, the new multi-agent market model proposed in this study is consistent with some previous studies.
Our model covers a wide range of models, including some previous works, because the functions of each trader are not defined specifically, especially in AI traders.
For example, Chiarella {\it et al.} \cite{CI2002} proposed a model based on fundamental trends and noise traders.
Because AI traders' definition can cover the trend predicting traders, it also can be covered by our theory.
Therefore, our microfoundations could be applicable to a wide range of models.

Second, this study newly achieves the microfoundation of the GARCH model, and our proposed model has been proven valid through theoretical analysis and multi-agent simulations. Previously, the GARCH model was purely empirical, lacking theoretical validation. However, through our microfoundation, we have provided a theoretical basis for the GARCH model. This represents significant progress in the field of financial markets, as the GARCH model is widely used in various applications.

Third, our microfoundation revealed many insights into the previously unexplained parameters in the GARCH model. Through our microfoundations, these parameters can now be explained using parameters related to micromechanisms in financial markets. This provides a clearer understanding of market dynamics.

According to our results, if more fundamental traders exist in a market, volatility clustering appears much more.
If the ratio of fundamental traders $p_1$ increased, $\beta$ in the GARCH model would increase according to theorem \ref{thm1}.
Originally, $\beta$ is the autocorrelation coefficient for volatility according to equation \ref{GARCH}.
Therefore, a high ratio of fundamental traders $p_1$ means that the strength of the volatility clustering becomes high.

Moreover, when more AI traders exist in a market, a more significant response to the shock would happen, and overshooting tends to happen more.
If the ratio of AI traders $p_2$ increased, $\alpha$ in the GARCH model would increase according to theorem \ref{thm1}.
Originally, $\alpha$ is the coefficient for the square of diffusion term $u_{t-1}^2$ in the equation \ref{GARCH}.
Therefore, if there exist more AI traders in markets, the response to the exogenous diffusion, i.e., market shocks, becomes bigger, and the market turbulence becomes bigger than the expected market shocks.

Table \ref{tbl:market_type} summarizes our results derived from our microfoundations regarding the different trader types in the market.

\section{Conclusion}
In this study, we proposed a novel multi-agent market model incorporating fundamental traders, noise traders, and AI traders. Our model, rooted in market micro-structure, reveals new microfoundations of the GARCH(1,1) model. Through both mathematical analysis and extensive simulation studies, we demonstrated that our multi-agent model can derive the GARCH formula and replicate key stylized facts observed in financial markets, such as volatility clustering and fat tails.

Our findings highlight the ability of the new microfoundations to describe market structures and the dynamics resulting from different proportions of trader types, including AI traders and their decision-making. Our model offers insights into how various trader interactions contribute to market behavior and volatility.

Future research directions include extending this framework to multivariate GARCH models\cite{bauwens2006multivariate} to capture more complex market dynamics. 
Moreover, this study employs a very simple and basic pricing model based on order imbalance as equation \ref{pricing}.
Therefore, we should also consider the other pricing models such as the double auction market mechanism as the same as the actual markets.
Finally, estimating the actual proportions of different types of traders in real-world markets remains a critical area for further investigation. 
These efforts will enhance our understanding of the impact of AI traders and other market participants, ultimately contributing to more stable and efficient financial markets.

\appendix
\section{Proof of lemma \ref{ref:lemma}}
\begin{proof}
Refer to \cite{tobin1965theory}. Since $\phi(\epsilon)$ does not depend on $\mu$ and $\sigma$, the expected utility can be expressed as a function of $ \mu $ and $ \sigma $:
\begin{align}
    \mathbb{E}[U] = \mathbb{E}_{\mu, \sigma}[U] = \int_{-\infty}^{\infty} U(\mu + \sigma \epsilon) \phi(\epsilon) d \epsilon
\end{align}
By differentiating this with respect to $ \sigma $ and showing that the derivative is non-positive:
\begin{align}
    \frac{\partial \mathbb{E}_{\mu, \sigma}[U]}{\partial \sigma} = \int_{-\infty}^{\infty} U'(\mu + \sigma \epsilon) \epsilon \phi(\epsilon) d \epsilon \leq 0
    \label{eq:sigma_deriv}
\end{align}
To prove this, let $ U'(\mu) = c $. Since $ U $ is risk-averse, $ U' $ is a decreasing function. Therefore, in the region where $ \epsilon < 0 $, $ c < U'(\mu + \sigma \epsilon) $, and in the region where $ \epsilon \geq 0 $, $ c \geq U'(\mu + \sigma \epsilon) $. Thus, the integral in \eqref{eq:sigma_deriv} is not greater than $ \int c \epsilon \phi(\epsilon) d \epsilon = 0 $, proving the lemma.
\end{proof}

\end{document}